\let\accentvec\vec
\documentclass[10pt]{llncs}

\let\vec\accentvec
\usepackage[utf8]{inputenc}
\usepackage[T1]{fontenc}
\usepackage{lmodern}
\usepackage{subfigure}
\usepackage{multirow}
\usepackage{cellspace}
\usepackage{slashbox}
\usepackage{amssymb}
\usepackage{amsmath}
\usepackage{amsfonts}
\usepackage{mathtools}
\usepackage{paralist}
\usepackage{latexsym}
\usepackage{listings}
\usepackage{relsize}
\usepackage[pdftex]{graphicx,color}
\usepackage{url}
\usepackage[normalem]{ulem}
 \newcommand\ForAuthors[1]
 {\par\smallskip                     
  \begin{center}
   \fbox
   {\parbox{0.9\linewidth}
    {\raggedright\sc--- #1}
   }
  \end{center}
  \par\smallskip                     %
 }
\newcommand{\comment}[1]{}

\newcommand{\assale}[1]{\textbf{{\color{green}#1}}}
\newcommand{\ploc}[1]{\textbf{{\color{blue}#1}}}

\newcommand{\K}{\mathcal K}

\newcommand{\nn}{\mathbb N}
\newcommand{\rr}{\mathbb R}
\newcommand{\rd}{\rr^d}

\def\norm#1{\mbox{$\| #1 \|$}}

\newcommand{\well}{well-representative }
\newcommand{\prop}[1]{\mathcal{P}_{#1}}

\newcommand{\pws}{\mathcal{S}}

\def\sizesmallfig{0.19}

\newcommand{\xlive}{X^0}

\newcommand{\rea}{\mathfrak{R}}

\newcommand{\state}{\mathcal{X}}
\newcommand{\laws}{\mathcal{L}}
\newcommand{\xin}{X^{\mathrm{in}}}
\newcommand{\ind}{\mathcal{I}}

\begin{document}
\title{Property-based Polynomial Invariant Generation using
    Sums-of-Squares Optimization
}
\author{Assal\'e Adj\'e\inst{1}$^{,a}$ and Pierre-Loïc Garoche\inst{1}$^{,a}$ and Victor Magron\inst{2}$^{,b}$}
\institute{
Onera, the French Aerospace Lab, France.\\
Universit\'e de Toulouse, F-31400 Toulouse, France.
\\
\email{firstname.lastname@onera.fr}
\and
Circuits and Systems Group, Department of Electrical and Electronic Engineering,
Imperial College London, South Kensington Campus, London SW7 2AZ, UK.\\
\email{v.magron@imperial.ac.uk}
}
\renewcommand{\thefootnote}{\alph{footnote}}
\footnotetext[1]{The author is supported by the RTRA /STAE Project BRIEFCASE and the ANR ASTRID
VORACE Project.}
\footnotetext[2]{The author is supported by EPSRC (EP/I020457/1) Challenging Engineering Grant.}
\maketitle
\begin{abstract}
  While abstract interpretation is not theoretically restricted to specific
  kinds of properties, it is, in practice, mainly developed to compute linear
  over-approximations of reachable sets, aka. the collecting semantics of the
  program. The verification of user-provided properties is not easily compatible
  with the usual forward fixpoint computation using numerical abstract domains.

  We propose here to rely on sums-of-squares programming to characterize a property-driven 
  polynomial invariant. 
  This invariant generation can be guided
  by either boundedness, or in contrary, a given
  zone of the state space to avoid. 

  While the target property is not necessarily
  inductive with respect to the program semantics, our method identifies a
  stronger inductive polynomial invariant using numerical optimization. Our
  method applies to a wide set of programs: a main while loop composed of a 
  disjunction (if-then-else) of polynomial updates e.g. piecewise polynomial 
  controllers. It has been evaluated on various programs.

\end{abstract}


\section{Introduction}

With the increased need for confidence in software, it becomes more than ever
important to provide means to support the verification of specification of
software. Among the various formal verification methods to support these
analysis, a first line of approaches, such as deductive methods or SMT-based
model checking, provide rich languages to support the expression of the
specification and then try to discharge the associate proof obligation using
automatic solvers. The current state of the art of these solvers is able to
manipulate satisfiability problems over linear arithmetics or restricted
fragments of non linear arithmetics. Another line of approaches, such as static
analysis also known as abstract interpretation, restricts, a priori, the kind of
properties considered during the computation: these methods typically perform
interval arithmetic analysis or rely on convex polyhedra computations. In
practice this second line of work seems more capable of manipulating and
generating numerical invariants through the computation of inductive
invariants, while the first line of approaches hardly synthesize these required
invariants through satisfiability checks.

However, when it comes to more than linear properties, the state of the art is
not well developed. In the early 2000s, ellipsoid
analyses~\cite{DBLP:conf/esop/Feret04}, similar to restricted cases of Lyapunov
functions, were designed to support the study of a family of Airbus
controllers. This exciting result was used to provide the analysis of absence of
runtime errors but could hardly be adapted to handle more general user provided
specifications for polynomial programs.

However proving polynomial inequalities is
NP-hard and boils down to show that the infimum of a given polynomial is
nonnegative.  Still, one can obtain lower bounds of such infima by decomposing
certain nonnegative polynomials into sums-of-squares (SOS).  This actually leads
to solve hierarchies of semidefinite relaxations, introduced by Lasserre
in~\cite{Las01moments}. Recent advances in semidefinite programming allowed to
extensively apply these relaxations to various fields, including parametric
polynomial optimization, optimal control, combinatorial optimization, {\em etc.}
(see e.g.~\cite{Parrilo2003relax,LaurentSurvey} for more details).

While these approaches were
mentioned a decade ago in~\cite{CousotSDP} and mainly applied to termination
analysis, they hardly made their way through the software verification
community to address more general properties.

\if{
\ploc{Meme si les
  realted works sont a la fin, est-ce qu'on ne pourrait pas mentionner
  brievement le VMCAI05 de cousot ici? Par ex: }
}\fi
\if{
Template domains were introduced by Sankaranarayanan et al.~\cite{Sriram1}, see also~\cite{Sriram2}.
The latter authors only considered a finite set of linear templates and did not provide an automatic method 
to generate templates. 
Linear template domains were generalized to nonlinear quadratic cases by Adj\'e et 
al. in~\cite{DBLP:journals/corr/abs-1111-5223,DBLP:conf/esop/AdjeGG10}, where the authors used 
in practice quadratic Lyapunov templates for affine arithmetic programs. These templates are again not automatically 
generated. 
}\fi


\paragraph{Contributions.}
Our contribution allows to analyze high level properties defined as
a sublevel set of polynomials functions, i.e. basic semialgebraic
sets. This class of properties is rather large: it ranges
from boundedness properties to the definition of a bad region of the state space
to avoid.
While these properties, when they hold, are meant to be invariant, i.e. they
hold in each reachable state, they are
not necessarily inductive.  Our approach rely on the computation of a stronger inductive property
using SOS programming. This stronger property is proved inductive on
the complete system and, by construction, implies the target property specified
by the user.
\if{
We develop our analysis on discrete-time piecewise polynomial controllers,
capturing a wide class of critical programs, as typically found in current
 embedded systems such as aircraft. The current analysis is defined,
without loss of generality, at the system level. It can be easily applied at
code level on controllers; especially on the ones combining linear controllers,
saturation, anti-windup and other non linear constructs.
}\fi
We develop our analysis on discrete-time piecewise polynomial systems,
capturing a wide class of critical programs, as typically found in current
embedded systems such as aircrafts. 

\paragraph{Organization of the paper.}The paper is organized as follows. In Section~\ref{sem-functional}, we present the programs that we want to analyze and their representation as piecewise polynomial discrete-time systems. 
Next, we recall in Section~\ref{sec:invariant} the collecting semantics that we use and introduce the polynomial optimization problem providing inductive invariants based on target polynomial properties. 
Section~\ref{sec:sos} contains the main contribution of the paper, namely how to compute effectively such invariants with SOS programming. 
Practical computation examples are provided in Section~\ref{bench}. Finally, we explain in Section~\ref{sec:templates} how to derive template bases from generated invariants.


\section{Polynomial programs and piecewise polynomial discrete-time systems}
\label{sem-functional}
In this section, we describe the programs which are considered in this paper and we explain how to analyze them through their representation as piecewise polynomial discrete-time dynamical systems. 

We focus on programs composed of a single loop with a possibly complicated switch-case type loop body.
Moreover we suppose without loss of generality that the analyzed programs are written in Static Single Assignment (SSA) form, that is each variable is initialized at most once. 
\paragraph{Definitions.}
We recall that a function $f$ from $\rd$ to $\rr$ is a polynomial if and only if there exists $k\in\nn$, a family 
$\{c_\alpha\mid \alpha=(\alpha_1,\ldots,\alpha_d)\in\nn^d,\ |\alpha|=\alpha_1+\ldots+\alpha_d\leq k\}$ such that for all $x\in\rd$, $f(x)=\sum_{|\alpha|\leq k} c_{\alpha} x_1^{\alpha_1}\ldots x_d^{\alpha_d}$. By extension a function $f:\rd\mapsto \rd$ is a polynomial 
if and only if all its coordinate functions are polynomials. Let $\rr[x]$ stands for the set of $d$-variate polynomials.

In this paper, we consider assignments of variables using only {\em parallel polynomial assignments} $(x_1,\ldots,x_d)=T(x_1,\ldots,x_d)$ where $(x_1,\ldots,x_d)$ is the vector of the program variables. Tests are either weak polynomial inequalities $r(x_1,\ldots,x_d) \leq 0$ or strict polynomial inequalities $r(x_1,\ldots,x_d) < 0$. We assume that assignments are polynomials from $\rd$ to $\rd$ and test functions are polynomials from $\rd$ to $\rr$. In the program syntax, the notation $\ll$ will be either $\verb+<=+$ or $\verb+<+$. The form of the analyzed program is described in Figure~\ref{programstyle}. 
\begin{figure}[!ht]
\begin{center}
\begin{tabular}{|c|}
\hline
\if{
\begin{lstlisting}[mathescape=true]
x $\in$ $\xin$;
while ($r_1^0$(x)$\ll$0 and ... and $r_{n_0}^0$(x)$\ll$0){
  if($r_1^1$(x)$\ll$0){ 
     $\vdots$
     if($r_{n_1}^1$(x)$\ll$0){
        x = $T^1$(x);
     }   
     else{
        $\vdots$
        if($r_{n_i}^i$(x)$\ll$0){
           x = $T^i$(x);   
        }   
     }
  else{
      $\vdots$
  } 
}
\end{lstlisting}
}\fi
\begin{lstlisting}[mathescape=true]
x $\in$ $\xin$;
while ($r_1^0$(x)$\ll$0 and ... and $r_{n_0}^0$(x)$\ll$0){
  case ($r_1^1$(x)$\ll$0 and ... and $r_{n_1}^1$(x)$\ll$0): x = $T^1$(x);
  case ...
  case ($r_i^1$(x)$\ll$0 and ... and $r_{n_i}^1$(x)$\ll$0): x = $T^i$(x); 
}
\end{lstlisting}
\\
\hline
\end{tabular}
\end{center}
\if{
\ploc{Victor, tu ne veux pas plutot avoir qqchose qui ressemble à un pattern
  matching?\\
  $case r^1_1(x)\ll$ 0 and ... and $r^1_{n_1}(x)\ll:\\~\qquad x = T^1(x);\\
  case ...\\
  case  r^k_1(x)\ll$ 0 and ... and $r^k_{n_k}(x)\ll:\\~\qquad x = T^k(x);\\
  ...$\\
  Ca evite d'avoir ces cascades de if mais aussi le else final qu'en pratique on
  aura jamais. Est-ce qu'il ne faut pas preciser que I est fini? Ca n'a pas trop
de sens d'avoir un programme avec un nombre infini de partitions et donc d'updates.}
}\fi
\caption{One-loop programs with nested conditional branches}
\label{programstyle}
\end{figure}

A set $C\subseteq \rd$ is said to be a \emph{basic semialgebraic set} if there exist $g_1,\ldots,g_m \in \rr[x]$ such that $C=\{x\in \rd\mid g_j(x)\ll 0, \forall\, j=1,\ldots,m\}$, where $\ll$ is used to encode either a strict or a weak inequality.

As depicted in Figure~\ref{programstyle}, an update $T^i : \rd \to \rd$ of the $i$-th condition branch is executed if and only if the conjunction of tests $r_j^i(x)\ll 0$ holds. In other words, the variable $x$ is updated by $T^i(x)$ if the current value of $x$ belongs to the basic semialgebraic set 
\begin{equation}
\label{semialgebraic}
X^i:= \{x \in \rd \, | \,  \forall j = 1,\dots,n_i, \ r_j^i(x)\ll 0 \}\enspace.           
\end{equation}
\paragraph{Piecewise Polynomial Systems.} 
Consequently, we interpret programs as \emph{constrained piecewise polynomial discrete-time dynamical systems} (PPS for short). 
The term \emph{piecewise} means that there exists a partition $\{X^i,i\in \ind\}$ of $\rd$ such that for all $i\in \ind$, 
the dynamics of the system is represented by the following relation, for $k\in\nn$:
\begin{equation}
\label{pws}
\text{if } x_k\in X^i\cap X^0,\ x_{k+1}=T^i(x_k) \,.
\end{equation}

We assume that $\ind$ is finite and that the initial condition $x_0$ belongs to some compact basic semialgebraic set $\xin$. For the program, $\xin$ is the set where the variables are supposed to be initialized in. Since the test entry for the loop condition can be nontrivial, we add the term \emph{constrained} and $X^0$ denotes the set representing the conjunctions of tests for the loop condition. The iterates of the PPS are constrained to live in $X^0$: if for some step $k\in\nn$, $x_k\notin X^0$ then the PPS is stopped at this iterate with the terminal value $x_k$.

We define a partition as a family of nonempty sets such that:
\begin{equation}
\label{partition}
\bigcup_{i\in\ind} X^i=\rd,\ \forall\, i,j\in\ind,\ i\neq j, X^i\cap X^j\neq \emptyset \,.
\end{equation}
From Equation~\eqref{partition}, for all $k\in\nn^*$ there exists a unique $i\in\ind$ such that $x_k\in X^i$. A set $X^i$ can contain both strict and weak polynomial inequalities and characterizes the set of the $n_i$ conjunctions of tests polynomials $r_j^i$. Let $r^i=(r_1^i,\ldots,r_{n_i}^i)$ stands for the vector of tests functions associated to the set $X^i$. 
\if{
Moreover, for $X^i$, we denote by $r^{i,s}$ (resp. $r^{i,w}$) the part of $r^i$ corresponding to strict (resp. weak) inequalities. Finally, we obtain the representation of the set $X^i$ given by Equation~\eqref{semialgebraic}:
\begin{equation}
\label{semialgebraic}
X^i=\left\{x\in\rd \left| r^{i,s}(x) < 0,\ r^{i,w}(x)\leq 0\right\}\right. \,.
\end{equation}
We insist on the notation: $y< z$ (resp. $y_l<z_l$) means that for all coordinates $l$, $y_l<z_l$ (resp. $y_l\leq z_l$).
}\fi
We suppose that the basic semialgebraic sets $\xin$ and $X^0$ also admits the representation given by Equation~\eqref{semialgebraic} and we denote by $r^0$ the vector of tests polynomials $(r_1^0,\ldots,r_{n_0}^0)$ and by $r^{\mathrm{in}}$ the vector of test polynomials $(r_1^{\mathrm{in}},\ldots,r_{n_{\mathrm{in}}}^{\mathrm{in}})$. 
\if{
We also decompose $r^0$ and $r^{\mathrm{in}}$ as strict and weak inequality parts denoted respectively by $r^{0,s}$, $r^{0,w}$, $r^{\mathrm{in},s}$ and $r^{\mathrm{in},w}$.
}\fi
To sum up, we give a formal definition of PPS.
\begin{definition}[PPS]
\label{pwsdef}
A constrained polynomial piecewise discrete-time dynamical system (PPS) is the quadruple $(\xin,X^0,\state,\laws)$ with:
\begin{itemize}
\item $\xin\subseteq \rd$ is the compact basic semialgebraic set of the possible initial conditions;
\item $X^0\subseteq \rd$ is the basic semialgebraic set where the state variable lives;
\item $\state:=\{X^i, i\in \ind\}$ is a partition as defined in Equation~\eqref{partition};
\item $\laws:=\{T^i, i\in\ind\}$ is the family of the polynomials from $\rd$ to $\rd$, w.r.t. the partition $\state$ satisfying Equation~\eqref{pws}.
\end{itemize}   
\end{definition}
 
From now on, we associate a PPS representation to each program of the form described at Figure~\ref{programstyle}. 
Since a program admits several PPS representations, we choose one of them, but this arbitrary choice does not change the results 
provided in this paper. 
In the sequel, we will often refer to the running example described in Example~\ref{running}.
\begin{example}[Running example]
\label{running}
The program below involves four variables and contains 
an infinite loop with a conditional branch in the loop body.
The update of each branch is polynomial. The parameters $c_{i j}$ (resp.  $d_{i j}$) are given parameters.
During the analysis, we only keep the variables $x_1$ and $x_2$ since 
$oldx_1$ and $oldx_2$ are just memories.
\begin{center}
\begin{tabular}{c}
\begin{lstlisting}[mathescape=true]
$x_1, x_2\in [a_1, a_2] \times [b_1, b_2]$;
$oldx_1$ = $x_1$;
$oldx_2$ = $x_2$;
while (-1 <= 0){
  $oldx_1$ = $x_1$;
  $oldx_2$ = $x_2$;
  case : $oldx_1$^2 + $oldx_2$^2 <= 1 : 
      $x_1$ = $c_{11}$ * $oldx_1$^2 + $c_{11}$ * $oldx_2$^3;
      $x_2$ = $c_{21}$ * $oldx_1$^3 + $c_{22}$ * $oldx_2$^2;
  case :  -$oldx_1$^2 - $oldx_2$^2 < -1     
      $x_1$ = $d_{11}$ * $oldx_1$^3 + $d_{12}$ * $oldx_2$^2;
      $x_2$ = $d_{21}$ * $oldx_1$^2 + $d_{22}$ * $oldx_2$^2;
  } 
}
\end{lstlisting}
\end{tabular}
\end{center}
The associated PPS corresponds to the quadruple $(\xin,X^0,\{X^1,X^2\},\{T^1,T^2\})$, where the set of initial conditions is:
\[ 
\xin= [a_1, a_2] \times [b_1, b_2] \,,
\]
the system is not globally constrained, i.e. the set $X^0$ in which the variable $x=(x_1,x_2)$ lies is:
\[
X^0=\rd \,,
\]
the partition verifying Equation~\eqref{partition} is:
\[
X^1=\{x\in\rr^2\mid x_1^2+x_2^2\leq 1\},\quad X^2=\{x\in\rr^2\mid -x_1^2-x_2^2< -1\} \,,
\]
and the polynomials relative to the partition $\{X^1,X^2\}$ are:
\[
T^1(x)=\left(\begin{array}{c}
c_{1 1} x_1^2 + c_{1 2} x_2^3\\
c_{2 1} x_1^3 + c_{2 2} x_2^2
\end{array}\right)
\text{ and }  
T^2(x)=\left(\begin{array}{c}
d_{1 1} x_1^3 + d_{1 2} x_2^2\\
d_{2 1} x_1^2 + d_{2 2} x_2^2
\end{array}\right)\enspace.
\]
\end{example}


\section{Program invariants as sublevel sets}
\label{sec:invariant}

The main goal of the paper is to decide automatically if a given property holds
for the analyzed program, i.e. for all its reachable states. We are interested in
numerical properties and more precisely in properties on the values taken by the
$d$-uplet of the variables of the program. Hence, in our point-of-view, a
property is just the membership of some set $P\subseteq \rd$. In particular, we
study properties which are valid after an arbitrary number of loop
iterates. Such properties are called \emph{loop invariants} of the
program. Formally, we use the PPS representation of a given program and we say
that $P$ is a loop invariant of this program if:
\[
\forall\, k\in\nn,\ x_k\in P \,,
\] 
where $x_k$ is defined at Equation~\eqref{pws} as the state variable at step
$k\in\nn$ of the PPS representation of the program. 
Our approach addresses any
  property expressible as a polynomial level set property. This section defines
  formally these notions and develop our approach: synthesize a property-driven
  inductive invariant. 
\subsection{Collecting Semantics as postfixpoint characterization}
Now, let us consider a program of the form described in
Figure~\ref{programstyle} and let us denote by $\pws$ the PPS representation of
this program. The set $\rea$ of \emph{reachable values} is the set of all
possible values taken by the state variable along the running of $\pws$. We
define $\rea$ as follows:
\begin{equation}
\label{reachable}
\rea = \bigcup_{k \in \nn} T_{|_{\xlive}}^k(\xin)
\end{equation}
where $T_{|_{\xlive}}$ is the restriction of $T$ on $\xlive$ and
$T_{|_{\xlive}}$ is not defined outside $\xlive$.
To prove that a set $P$ is a loop invariant of the program is equivalent to prove that $\rea\subseteq P$. 
We can rewrite $\rea$ inductively:
\begin{equation}
\label{auxsemantics}
\rea=\xin\cup \bigcup_{i\in\ind} T^i\left(\rea\cap X^i\cap \xlive\right)\,.
\end{equation}
Let us denote by $\wp(\rd)$ the set of subsets of $\rd$ and introduce the map $F: \wp(\rr^d) \rightarrow \wp(\rr^d)$ defined by:
\begin{equation}
\label{transferfunctional}
F(C)
=\xin \cup \bigcup_{i\in\ind} T^i\left(C\cap X^i\cap X^0\right)
\end{equation}
We equip $\wp(\rr^d)$ with the partial order of inclusion. The infimum is understood in this sense i.e. as the greatest lower bound with respect to this order. The smallest fixed point problem is:
\begin{equation*}
\inf \left\{C\in\wp(\rr^d)\mid C=F(C) \right\}\,.
\end{equation*}
It is well-known from Tarski's theorem that the solution of this problem exists, is unique and in this case, it corresponds to $\rea$. Tarski's theorem also states that $\rea$ is the smallest solution of the following Problem:
\begin{equation*}
\inf \left\{C\in\wp(\rr^d)\mid F(C)\subseteq C\right\} \,.
\end{equation*}

\if{
We warn the reader that the construction of $F$ is completely determined by the
data of the PPS $\pws$. But for the sake of conciseness, we do not make it
explicit on the notations\ploc{~je ne vois pas ce que tu/vous voulez dire ici}.
}\fi 
Note also that the map $F$ corresponds to a standard transfer function (or
collecting semantics functional) applied to the PPS representation of a program.
We refer the reader to~\cite{CC77} for a seminal presentation of this approach.


To prove that a subset $P$ is a loop invariant, it suffices to show that $P$ satisfies $F(P)\subseteq P$. In this case, such $P$ is called 
\emph{inductive invariant}.
 
\subsection{Considered properties: sublevel properties $\prop{\kappa, \alpha}$}

In this paper, we consider special properties: those that are encoded with
sublevel sets of a given polynomial function.
\begin{definition}[Sublevel property]
\label{funproperty}
Given a polynomial function $\kappa \in \rr[x]$
and $\alpha \in \rr \cup \{ +\infty \}$, we define the sublevel property $\prop{\kappa, \alpha}$ as follows:
\[
\prop{\kappa, \alpha}:= \{x\in\rd\mid \kappa(x)\ll \alpha\} \,.
\]
where $\ll$ denotes $\leq$ when $\alpha \in \rr$ and denotes $<$ for
$+\infty$. The expression $\kappa(x) < +\infty$ expresses the boundedness of $\kappa(x)$ without
providing a specific bound $\alpha$.
\end{definition}

\begin{example}[Sublevel property examples]

\noindent\emph{Boundedness.} When one wants to bound the reachable values of a
system, we can try to bound the $l_2$-norm of the system: $\prop{\|\cdot \|_2^2,
\infty}$ with $\kappa(x) = \|x \|_2^2$. The use of $\alpha = \infty$ does not
impose any bound on $\kappa(x)$.

\noindent\emph{Safe set.} Similarly, it is possible to check whether a specific
bound is matched. Either globally using the $l_2$-norm and a specific $\alpha$: $\prop{\|\cdot \|_2^2,
\alpha}$, or bounding the reachable values of each variable: $\prop{\kappa_i,
\alpha_i}$ with $\kappa_i : x \mapsto x_i$ and $\alpha_i \in \rr$.

\noindent\emph{Avoiding bad regions.} If the bad region can be encoded as a sublevel
property $k(x) \leq 0$ then its negation $-k(x) \leq 0$ characterize the
avoidance of that bad zone. Eg. if one wants to prove that the square norm of
the program variables is always greater than 1, then we can consider the
property $\prop{\kappa, \alpha}$ with $\kappa(x)=1- \|x \|_2^2$ and $\alpha =
0$.  


\end{example}

A sublevel property is called \emph{sublevel invariant} when this property is a loop invariant.
This turns out to be difficult to prove loop invariant properties while considering directly $\rea$, thus we propose to find a more tractable over-approximation of $\rea$ for which such properties hold.


\if{
In the next, we are looking for a family of $k$ polynomials $\{p_1,\ldots,p_k\}$ such that:
\begin{equation}
\label{eqfondamentale}
\rea\subseteq \bigcap_{i=1}^k \{x\in\rd\mid p_i(x)\leq 0\} \subseteq \prop{\kappa,\alpha}
\end{equation}
}\fi

\subsection{Approach: compute a $\prop{\kappa, \alpha}$-driven inductive invariant $P$}
In this subsection, we explain how to compute a $d$-variate polynomial $p \in
\rr[x]$ and a bound $w \in \rr$, such that the polynomial sublevel sets $P := \{x\in\rd\mid p(x)\leq 0\}$ and $\prop{\kappa,w}$ satisfy:
\begin{equation}
\label{eqfondamentale}
\rea \subseteq P \subseteq \prop{\kappa,w} \subseteq \prop{\kappa,\alpha} \,.
\end{equation}
The first (from the left) inclusion forces $P$ to be valid for the whole reachable values set. The second inclusion constraints all elements of $P$ to satisfy the given sublevel property for a certain bound $w$. The last inclusion requires that the bound $w$ is smaller than the desired level $\alpha$. When $\alpha=\infty$, any bound $w$ ensures the sublevel property.  
%

Now, we derive sufficient conditions on $p$ and $w$ to satisfy Equation~\eqref{eqfondamentale}. We decompose the 
problem in two parts. To satisfy the first inclusion, i.e.~ensure that $P$
is a loop invariant, it suffices to guarantee that $F(P) \subseteq P$, namely that $P$ is 
an inductive invariant.
\if{
\[
F\left(\bigcap_{j=1}^k \{x\in\rd\mid p_j(x)\leq 0\}\right)\subseteq \bigcap_{j=1}^k \{x\in\rd\mid p_j(x)\leq 0\}
\]
}\fi
Using Equation~\eqref{auxsemantics}, $P$ is an inductive invariant if and only if:
\[
\xin \cup \bigcup_{i\in\ind} T^i\left(P \cap X^i\cap X^0\right)\subseteq P \,,
\]
or equivalently:
\begin{equation}
\label{eqinvarianttemplates}
\left\{
\begin{array}{l}
\displaystyle{\xin \subseteq P} \,,\\
\displaystyle{
\forall\, i\in\ind,\ 
T^i\left(P \cap X^i\cap X^0\right)\subseteq P  \,.
}
\end{array}
\right.
\end{equation}
%
Thus, we obtain:
\if{
\[
\sup_{x\in \xin} p(x)\leq 0
\text{ and }
\forall\, i\in\ind,\ \sup_{\substack{p_l(x)\leq 0,\ \forall\, l=1,\ldots, k\\ x\in \xlive,\ x\in X^i}} p_j(T^i(x))\leq 0
\]
}\fi
\begin{align}
\label{eq:inv1}
\left\{
\begin{array}{lrl}
& p(x) \leq 0 \,, & \quad \forall x \in \xin \,,\\
\forall\, i\in\ind \,, &  p \, (T^i(x))\leq 0  \,, & \quad \forall x \in P \cap X^i\cap X^0  \,.
\end{array}
\right.
\end{align}
\comment{
\ploc{Je ne vois pas la justification sur le decoupage en P et
  $P_{k,w}$. Est ce uniquement pour le cas où alpha est laissé libre? Dans le
  paragraphe qui suit, j'expliquerai plus en detail les deux inegalités (i)
  since $P \subseteq Pkw$, then any element of P satisfies the constraint defining
  Pkw, i.e. $k(x) \leq w$. (ii) furthermore since Pkw is included in Pka, then $w
  \leq a$. We then have ... }
  \assale{Plus haut, ça suffit maintenant?}
}
Now, we are interested in the second and third inclusions at Equation~\eqref{eqfondamentale} that is the sublevel property satisfaction. The condition $P \subseteq \prop{\kappa, w} \subseteq  \prop{\kappa, \alpha}$ can be formulated as follows:
\begin{equation}
\label{eq:inv2}
\kappa(x) \leq w \leq \alpha \,, \quad \forall x \in P \,.\\
\end{equation}
We recall that we have supposed that $P$ is written as $\{x\in\rd\mid p(x)\leq 0\}$ where $p\in \rr[x]$.
Finally, we provide sufficient conditions to satisfy both~\eqref{eq:inv1} and~\eqref{eq:inv2}. Consider the following optimization problem:
\begin{align}
\label{eq:invsufficient}
\left\{
\begin{array}{rll}
\inf_{p \in \rr[x], w \in \rr} & \quad w  \,, & \\			 
\text{s.t.}  & \quad p(x) \leq 0 \,, & \quad \forall x \in \xin \,,\\
                  & \quad \forall\, i\in\ind \,,  p \, (T^i(x))\leq p(x)  \,, & \quad \forall x \in X^i\cap X^0  \,,\\
                  & \quad \kappa(x) \leq w + p(x) \,, & \quad \forall x \in \rd \,.
\end{array} 
\right.
\end{align}
\comment{
\ploc{On retrouve directement les deux premieres contraintes de (9). Mais pour
  la derniere contrainte? Pourquoi ce terme en p(x) ?}
}
%

We remark that $\alpha$ is not present in Problem~\eqref{eq:invsufficient}. Indeed, since we minimize $w$, either there exists 
a feasible $w$ such that $w\leq \alpha$ and we can exploit this solution or such $w$ is not available and we cannot conclude. However, 
from Problem~\eqref{eq:invsufficient}, we can extract $(p,w)$ and in the case where the optimal bound $w$ is greater than $\alpha$, we could use this solution with another method such as policy iteration~\cite{pisos}.  

\begin{lemma}
\label{lemma:invsufficient}
Let $(p, w)$ be any feasible solution of Problem~\eqref{eq:invsufficient} with
$w\leq \alpha$ or $w<\infty$ in the case of $\alpha=\infty$.
Then $(p, w)$ satisfies both~\eqref{eq:inv1} and~\eqref{eq:inv2} with $P := \{x\in\rd\mid p(x)\leq 0\}$. Finally, $P$ and $\prop{\kappa,w}$ satisfy Equation~\eqref{eqfondamentale}.
\end{lemma}
In practice, we rely on sum-of-squares programming to solve a relaxed version of Problem~\eqref{eq:invsufficient}.



\section{Sums-of-Squares Programming for Invariant Generation}
\label{sec:sos}
We first recall some basic background about sums-of-squares certificates for polynomial optimization.
Let $\rr[x]_{2m}$ stands for the set of polynomials of degree at most $2 m$ and
$\Sigma[x] \subset \rr[x]$ be the cone of sums-of-squares (SOS) polynomials,
that is $\Sigma[x] := \{\,\sum_i q_i^2, \, \text{ with } q_i \in \rr[x] \,\}$. Our work will use the simple fact that 
for all $p\in\Sigma[x]$, then $p(x)\geq 0$ for all $x\in\rd$ i.e. $\Sigma[x]$ is a restriction of the set of the nonnegative polynomials.
%
For $q \in \rr[x]_{2m}$,  finding a SOS decomposition $q = \sum_i q_i^2$ valid
over $\rd$ is equivalent to solve the following matrix linear feasibility problem:
\begin{align}
\label{eq:sdp}
q(x) = b_m(x)^T \, Q \, b_m(x) \,, \quad \forall x \in \rd, \,
\end{align}
where $b_m(x) := (1, x_1, \dots, x_d, x_1^2,x_1 x_2,\dots, x_d^m)$ (the vector of all monomials in $x$ up to degree $m$) and $Q$ being a {\em semidefinite positive} matrix (i.e. all the eigenvalues of $Q$ are nonnegative). The size of $Q$ (as well as the length of $b_m$) is ${d + m \choose d}$. 

\begin{example}
 consider the bi-variate polynomial $q(x) := 1 + x_1^2 - 2 x_1 x_2 + x_2^2$. With $b_1 (x) = (1, x_1, x_2)$, one looks for a semidefinite positive matrix $Q$ such that the polynomial equality $q(x) = b_1(x)^T \, Q \, b_1(x)$ holds for all $x \in \rr^2$. 
The matrix \[Q = 
\begin{pmatrix}
1 & 0 & 0\\
0 & 1 & -1\\
0 & -1 & 1
\end{pmatrix}\] satisfies this equality and has three nonnegative eigenvalues,
which are 0, 1, and 2, respectively associated to the three eigenvectors $e_0 :=
(0, 1, 1)^\intercal $, $e_1 := (1, 0, 0)^\intercal $ and $e_2 := (0, 1, -1)^\intercal$. Defining the matrices $L
:= (e_1 \, e_2 \, e_0)=\begin{psmallmatrix}1 & 0 & 0\\ 0 & 1 & 1\\0 & -1 & 1 \end{psmallmatrix}$ and $D = \begin{psmallmatrix} 1 & 0 & 0 \\ 0 & 2 & 0 \\ 0 & 0 & 0 \end{psmallmatrix}$, one obtains the decomposition $Q = L^\intercal \, D \, L$ and the equality $q(x) = (L \,
b_1(x))^T \, D \, (L \, b_1(x)) = \sigma (x) = 1 + (x_1 - x_2)^2$, for all $x
\in \rr^2$. The polynomial $\sigma$ is called a {\em SOS certificate} and
guarantees that $q$ is nonnegative.
\end{example}
In practice, one can solve the general problem~\eqref{eq:sdp} by using semidefinite programming (SDP) solvers (e.g. {\sc Mosek}~\cite{mosek}, SDPA~\cite{Yamashita10SDPA}). For more details about SDP, we refer the interested reader to~\cite{Vandenberghe94semidefiniteprogramming}.

\if{
inf p
s.t.  kappa  <= p  partout
       p o Ti   <= p  on Xi
               p  <= 0   on X0
}\fi

One way to strengthen the three nonnegativity constraints of Problem~\eqref{eq:invsufficient} is 
to consider the following {\em hierarchy} of SOS programs, parametrized by the
integer $m$ representing the half of the degree of $p$:
\begin{align}
\label{polsynthesis}
\left\{
\begin{aligned}
\inf_{p\in \rr[x]_{2m}, w\in\rr} & \quad w \enspace, \\			 
\text{s.t.}  
& \quad - p = \sigma_0 - \sum_{j=1}^{n_{\mathrm{in}}} \sigma_j r_j^{\mathrm{in}}  \enspace , \\
& \quad \forall\, i\in\ind,\ \displaystyle{p-p \circ T^i= \sigma^i - \sum_{j=1}^{n_i}\mu_j^i r_j^i - \sum_{j=1}^{n_0}\gamma_j^i r_j^0}  \enspace , \\
& \quad \displaystyle{w + p -\kappa = \psi} \enspace , \\
&\\
& \quad \forall\, j=1,\ldots, n_{\mathrm{in}} \enspace,\ \sigma_j\in\Sigma[x]\enspace ,\ \deg (\sigma_j r_j^{\mathrm{in}})  \leq 2m\enspace,\\
& \quad \sigma_0\in\Sigma[x]\enspace ,\ \deg (\sigma_0)  \leq 2m\enspace,\\
& \quad \forall\, i\in\ind \enspace ,\ \sigma^i\in \Sigma[x]\enspace ,\ \deg (\sigma^i) \leq 2 m \deg T^i \enspace,\\
& \quad \forall\, i\in\ind \enspace ,\ \forall\, j=1,\ldots, n_i \enspace,\ \mu_j^i\in\Sigma[x]\enspace ,\ \deg (\mu_j^i r_j^i)  \leq 2 m \deg T^i \enspace ,\\
& \quad \forall\, i\in\ind\enspace ,\ \forall\, j=1,\ldots, n_0\enspace ,\ \gamma^i\in \Sigma[x]\enspace ,\ \deg (\gamma_j^i r_j^0)  \leq 2 m \deg T^i  \enspace ,\\
& \quad \psi \in \Sigma[x] \enspace,\ \deg (\psi) \leq 2 m  \enspace .\\
\end{aligned} \right.
\end{align}

\begin{proposition}
\label{sosproposition}
For a given $m \in \nn$, let $(p_m, w_m)$ be any feasible solution of Problem~\eqref{polsynthesis}. Then $(p_m, w_m)$ is also a feasible solution of Problem~\eqref{eq:invsufficient}.
Moreover, if $w_m\leq \alpha$ then both $P_m := \{x \in \rd \mid p_m(x) \leq 0\}$ and $\prop{\kappa,w_m}$ satisfy Equation~\eqref{eqfondamentale}.
\end{proposition}
\begin{proof}
\label{funslemma}
The feasible solution $(p_m, w_m)$ is associated with SOS certificates ensuring that the three equality constraints of Problem~\eqref{polsynthesis} hold: $\{\sigma_0, \sigma_j\}$ is associated to the first one, $\{\sigma^i,\mu_j^i,\gamma_j^i\}$ is associated to the second one and $\psi$ is associated to the third one. The first equality constraint, namely
\[  - p_m(x) = \sigma_0(x) - \sum_{j=1}^{n_{\mathrm{in}}} \sigma_j(x) r_j^{\mathrm{in}}(x)  \,, \quad \forall x \in \rd  \,, \] 
implies that $\forall x \in \xin \,, p_m(x)\leq 0$. Similarly, one has $\forall  i\in\ind,  \forall x \in X^i\cap X^0 , p_m \, (T^i(x))\leq p_m(x)$ and $\forall x \in \rd, \kappa(x) \leq w_m + p_m(x)$. Then $(p_m, w_m)$ is a feasible solution of Problem~\eqref{eq:invsufficient}. The second statement comes directly from Lemma~\ref{lemma:invsufficient}.
\if{
Let $h,g_1,\ldots,g_n \in \rr[x]$.
If there exists $\mu\in \Sigma[x]^n$ such that
\[
\sum_{i=1}^n \mu_i(x)g_i(x))-h(x)\in\Sigma[x]\enspace ,
\]
then 
\[
\forall x \in \rd \, , (g_1(x)\leq 0\wedge \ldots\wedge g_n(x)\leq 0 \implies h(x)\leq 0) \, .
\]
}\fi
\end{proof}

\paragraph*{Computational considerations.} Define $t := \max\{\deg T^i, i\in\ind \}$. At step $m$ of this hierarchy, the number of SDP variables is proportional to $\binom{d + 2 m t}{d}$ and the number of SDP constraints is proportional to $\binom{d + m t}{d}$. Thus, one expects tractable approximations when the number $d$ of variables (resp. the degree $2 m$ of the template $p$) is small. However, one can handle bigger instances of Problem~\eqref{polsynthesis} by taking into account the system properties. For instance one could exploit sparsity as in~\cite{Waki06sumsof} by considering the variable sparsity correlation pattern of the polynomials $\{T^i,i\in\ind\},\{r_j^i, i\in\ind, j=1,\ldots,n_i\},\{r_j^0,j=1,\ldots,n_0\},\{r_j^{\mathrm{in}},j=1,\ldots,n_{\mathrm{in}}\}$ and $\kappa$. 


\section{Benchmarks}

Here, we perform some numerical experiments while solving Problem~\eqref{polsynthesis} (given in Section~\ref{sec:sos}) on several examples. 
In Section~\ref{benchbound}, we verify that the program of Example~\ref{running} satisfies some boundedness property. We also provide examples involving higher dimensional cases. Then, Section~\ref{benchsafe} focuses on other properties, such as checking that the set of variable values avoids an unsafe region. Numerical experiments are performed on an Intel Core i5 CPU ($2.40\, $GHz) with {\sc Yalmip} being interfaced with the SDP solver {\sc Mosek}. 
\label{bench}
\subsection{Checking boundedness of the set of variables values}
\label{benchbound}
\begin{example}
\label{ex:test}
Following Example~\ref{running}, we consider the constrained piecewise discrete-time dynamical system $\pws=(\xin,X^0,\{X^1,X^2\},\{T^1,T^2\})$  with $\xin = [0.9, 1.1] \times [0, 0.2] $, $X^0=\{x\in\rr^2\mid r^0(x)\leq 0\}$ with $r^0:x\mapsto -1$, 
$X^1=\{x\in\rr^2\mid r^1(x)\leq 0\}$ with $r^1:x\mapsto \norm{x}^2-1$, $X^2=\{x\in\rr^2\mid r^2(x)<0\}$
with $r^2=-r^1$ and $T^1:(x_1,x_2)\mapsto (c_{11}x_1^2+c_{12}x_2^3,c_{21}x_1^3+c_{22}x_2^2)$, 
 $T^2:(x_1,x_2) \mapsto (d_{11}x_1^3+d_{12}x_2^2,d_{21}x_1^2+d_{22}x_2^2)$. We are interested in showing that the boundedness property $\prop{\| \cdot \|_2^2,\alpha}$ holds for some positive $\alpha$. 
\end{example}
\begin{figure}[!ht]
\centering
\subfigure[$m = 3$]{
\includegraphics[scale=\sizesmallfig]{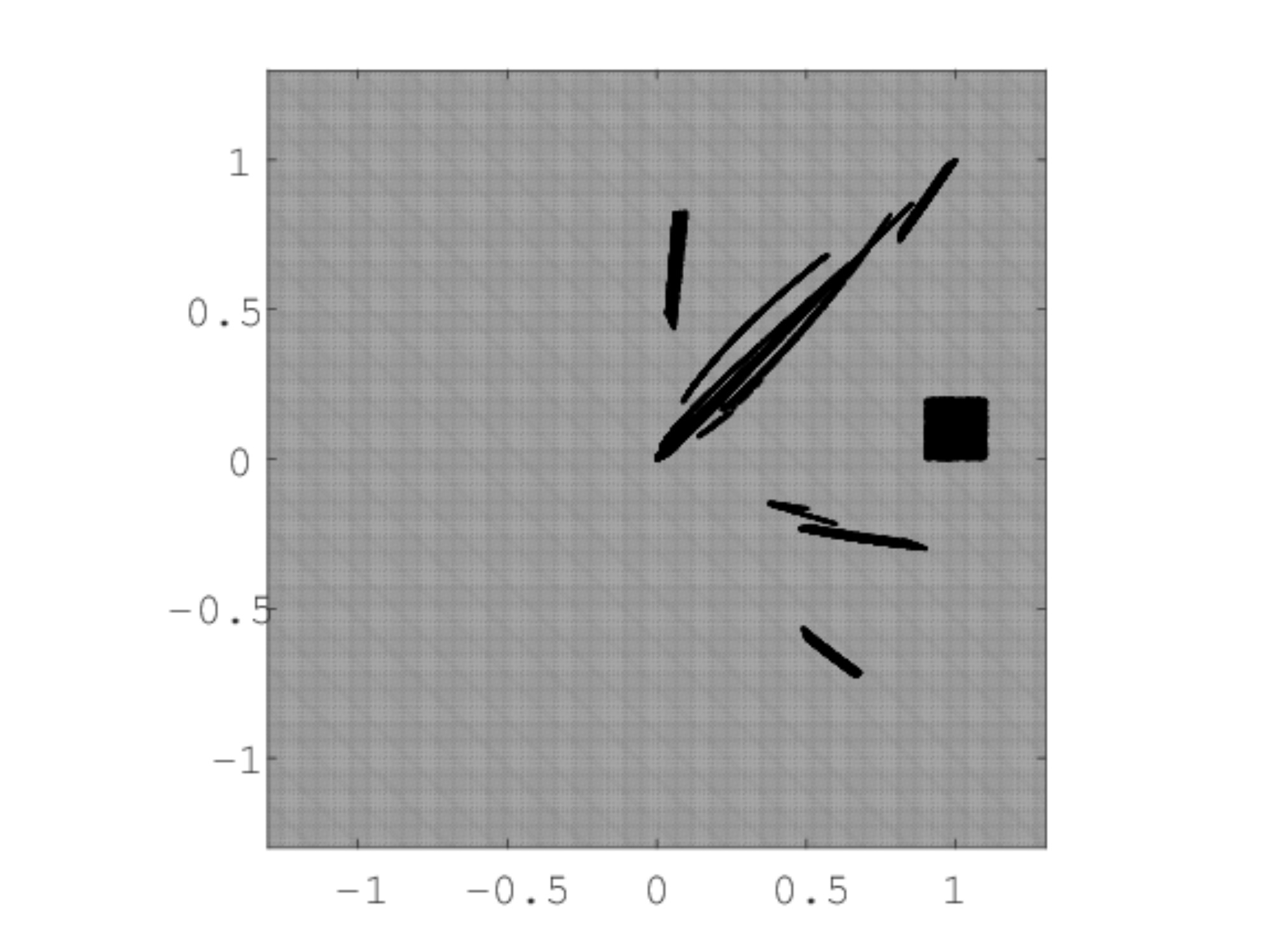}}
\subfigure[$m = 4$]{
\includegraphics[scale=\sizesmallfig]{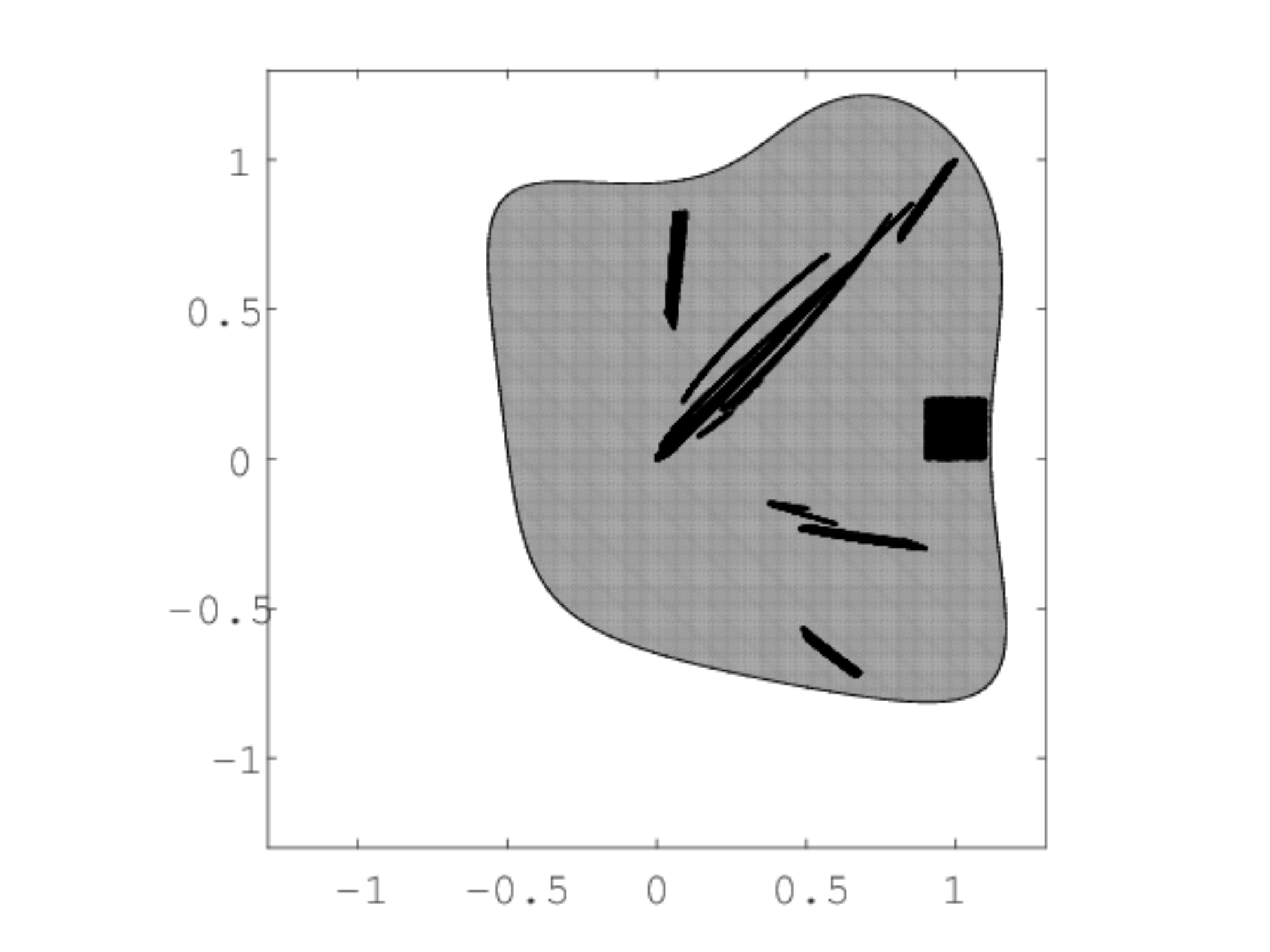}}
\subfigure[$m = 5$]{
\includegraphics[scale=\sizesmallfig]{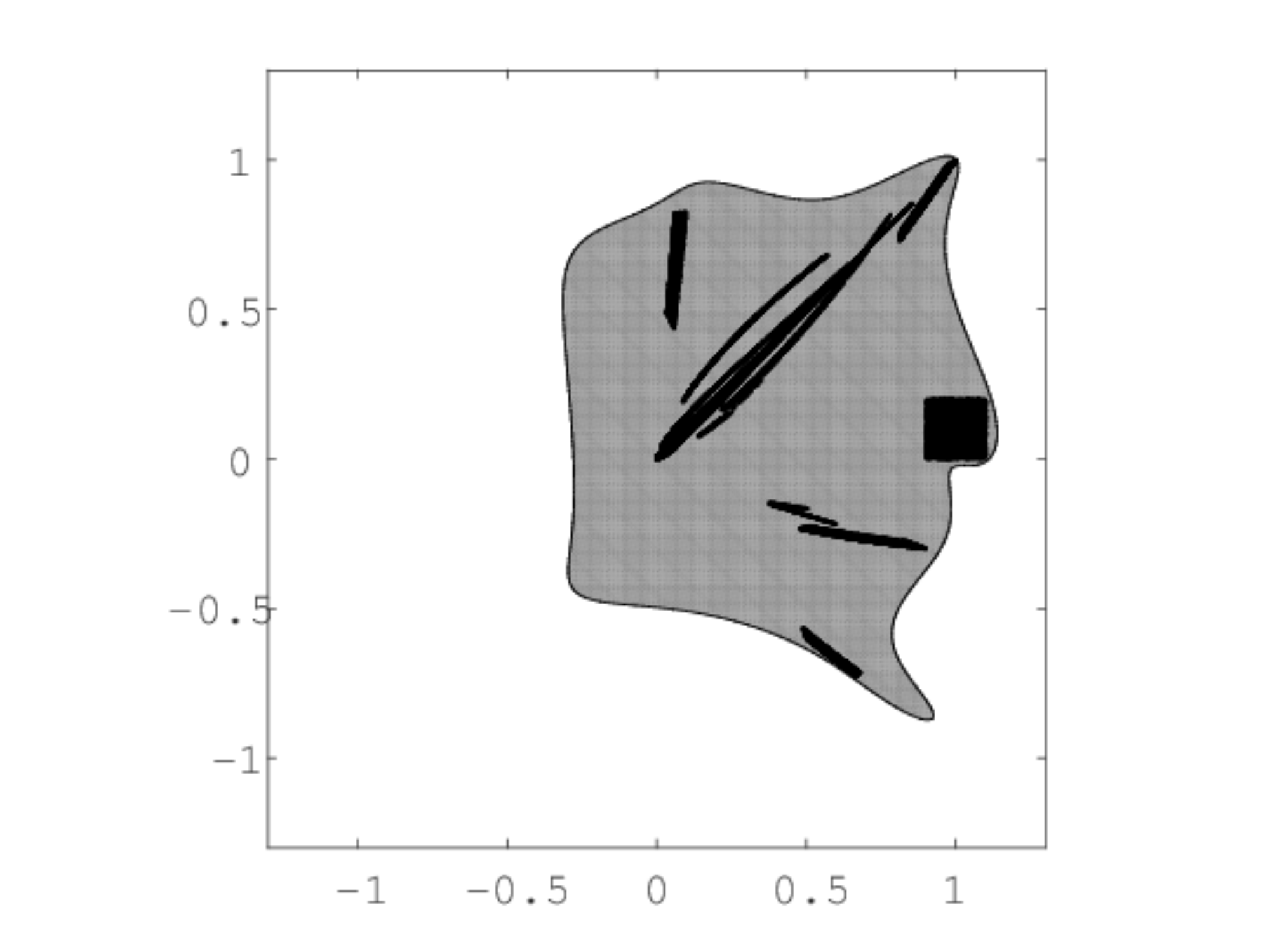}}
\caption{A hierarchy of sublevel sets $P_m$ for Example~\ref{ex:test}}	
\label{fig:test}
\end{figure}
Here we illustrate the method by instantiating the program of
Example~\ref{running} with the following input: $a_1 = 0.9$, $a_2 = 1.1$, $b_1 =
0$, $b_2 = 0.2$, $c_{11} = c_{12} = c_{21} = c_{22} = 1$, $d_{11} = 0.5$,
$d_{12} = 0.4$, $d_{21} = -0.6$ and $d_{22} = 0.3$.  We represent the possible
initial values taken by the program variables $(x_1, x_2)$ by picking uniformly
$N$ points $(x_1^{(i)}, x_2^{(i)}) \ (i = 1, \dots, N)$ inside the box $\xin =
[0.9, 1.1] \times [0, 0.2]$ (see the corresponding square of dots on
Figure~\ref{fig:test}). The other dots are obtained after successive updates of
each point $(x_1^{(i)}, x_2^{(i)})$ by the program of Example~\ref{running}. The
sets of dots in Figure~\ref{fig:test} are obtained with $N = 100$ and six
successive iterations.

At step $m = 3$, Program~\eqref{polsynthesis} yields a solution $(p_3, w_3) \in
\rr_6[x] \times \rr$ together with SOS certificates, which guarantee the
boundedness property, that is $x \in \rea\implies x \in P_3 := \{p_3 (x)
\leq 0 \} \subseteq \prop{\| \cdot \|_2^2,w_3} \implies \| x \|_2^2 \leq w_3$. One has $p_3(x) :=
-2.510902467-0.0050x_1-0.0148x_2+3.0998x_1^2-0.8037x_2^3
-3.0297x_1^3+2.5924x_2^2+1.5266x_1x_2-1.9133x_1^2x_2
-1.8122x_1x_2^2+1.6042x_1^4+0.0512x_1^3x_2-4.4430x_1^2
x_2^2-1.8926x_1x_2^3+0.5464x_2^4-0.2084x_1^5+0.5866x_1^4
x_2+2.2410x_1^3x_2^2+1.5714x_1^2x_2^3-0.0890x_1x_2^4-
0.9656x_2^5+0.0098x_1^6-0.0320x_1^5x_2-0.0232x_1^4x_2^2+
0.2660x_1^3x_2^3+0.7746x_1^2x_2^4+0.9200x_1x_2^5+0.6411
x_2^6$ (for the sake of conciseness, we do not display $p_4$ and $p_5$).

Figure~\ref{fig:test} displays in light gray outer approximations of the set of
possible values $X_1$ taken by the program of Example~\ref{ex:test} as follows:
(a) the degree six sublevel set $P_3$, (b) the degree eight sublevel set $P_4$
and (c) the degree ten sublevel set $P_5$. The outer approximation $P_3$ is
coarse as it contains the box $[-1.5, 1.5]^2$. However, solving
Problem~\eqref{polsynthesis} at higher steps yields tighter outer approximations
of $\rea$ together with more precise bounds $w_4$ and $w_5$ (see  the corresponding row in Table~\ref{table:benchcmp}).

We also succeeded to certify that the same property holds for higher dimensional programs, described in Example~\ref{ex:test3} ($d = 3$) and Example~\ref{ex:test4} ($d = 4$).
\begin{example}
\label{ex:test3}
Here we consider $\xin = [0.9, 1.1] \times [0, 0.2]^2 $, $r^0:x\mapsto -1$, 
$r^1: x \mapsto \| x \|_2^2-1$, $r^2=-r^1$, $T^1:(x_1,x_2,x_3)\mapsto 1/4 (0.8 x_1^2 + 1.4 x_2 - 0.5  x_3^2,   1.3  x_1 + 0.5  x_3^2,  1.4  x_2 + 0.8  x_3^2)$, 
 $T^2:(x_1,x_2, x_3) \mapsto 1/4(0.5  x_1 + 0.4  x_2^2,   -0.6  x_2^2 + 0.3  x_3^2,  0.5  x_3 + 0.4  x_1^2)$ and $\kappa : x \mapsto \| x \|_2^2$.
\end{example}
\begin{example}
\label{ex:test4}
Here we consider $\xin= [0.9, 1.1] \times [0, 0.2]^3 $, $r^0:x\mapsto -1$, 
$r^1: x \mapsto \| x \|_2^2-1$, $r^2=-r^1$, $T^1:(x_1,x_2,x_3, x_4)\mapsto 0.25 (0.8  x_1^2 + 1.4 x_2 - 0.5  x_3^2, 1.3  x_1 + 0.5,  x_2^2 - 0.8  x_4^2,         0.8  x_3^2 + 1.4 x_4,    1.3  x_3 + 0.5  x_4^2)$, 
 $T^2:(x_1,x_2, x_3, x_4) \mapsto 0.25 (0.5  x_1 + 0.4  x_2^2,   -0.6  x_1^2 + 0.3  x_2^2, 0.5  x_3 + 0.4  x_4^2,   -0.6  x_3 + 0.3  x_4^2)$ and $\kappa : x \mapsto \| x \|_2^2$.
\end{example}
Table~\ref{table:bench} reports several data obtained while solving Problem~\eqref{polsynthesis} at step $m$, ($2 \leq m \leq 5$), either for Example~\ref{ex:test}, Example~\ref{ex:test3} or Example~\ref{ex:test4}.
Each instance of Problem~\eqref{polsynthesis} is recast as a SDP program, involving a total number of ``Nb. vars'' SDP variables, with a SDP matrix of size ``Mat. size''. We indicate the CPU time required to compute the optimal solution of each SDP program with {\sc Mosek}.

The symbol ``$-$'' means that the corresponding SOS program could not be solved within one day of computation. These benchmarks illustrate the computational considerations mentioned in Section~\ref{sec:sos} as it takes more CPU time to analyze higher dimensional programs.
Note that it is not possible to solve Problem~\eqref{polsynthesis} at step $5$ for Example~\ref{ex:test4}. A possible workaround to limit this computational blow-up would be to exploit the sparsity of the system.

%
\renewcommand{\tabcolsep}{0.4cm}
\begin{table}[!ht]
\begin{center}
\caption{Comparison of timing results for Example~\ref{ex:test},~\ref{ex:test3} and~\ref{ex:test4}}
\begin{tabular}{c|c|cccc}
\hline
\multicolumn{2}{c|}{Degree $2 m$}
& 4 & 6 & 8 & 10
\\
\hline  
\multirow{2}{*}{Example \ref{ex:test}} & Nb. vars &  1513 & 5740 & 15705 & 35212 \\
& Mat. size & 368 & 802 & 1404 & 2174 \\
 ($d = 2$) & Time & $0.82 \, s$ & $1.35 \, s$ &  $4.00 \, s$ & $9.86 \, s$\\
\hline
\multirow{2}{*}{Example \ref{ex:test3}} & Nb. vars &  2115 & 11950 & 46461 & 141612\\
& Mat. size & 628 & 1860 & 4132 & 7764 \\
 ($d = 3$)& Time & $0.84 \, s$ & $2.98 \, s$ &  $21.4 \, s$ & $109 \, s$\\
\hline
\multirow{2}{*}{Example \ref{ex:test4}} & Nb. vars & 7202  & 65306 & 18480 & $-$\\
    & Mat. size & 1670 & 6622 & 373057 & $-$\\
 ($d = 4$) & Time & $2.85 \, s$ & $57.3 \, s$ &  $1534 \, s$ & $-$\\
\hline
\end{tabular}
\label{table:bench}
\end{center}
\end{table}
\renewcommand{\tabcolsep}{0.2cm}
\begin{table}[!ht]
\begin{center}
\caption{Hierarchies of bounds obtained for various properties}
\begin{tabular}{c|c|cccc}
\hline
\multicolumn{1}{c|}{Benchmark}
& $\kappa$ & $w_2$ & $w_3$ & $w_4$ & $w_5$
\\
\hline  
\multirow{1}{*}{Example \ref{ex:test}} & $\|\cdot\|_2^2$ & \multirow{1}{*}{639} & \multirow{1}{*}{17.4} & \multirow{1}{*}{2.44} & \multirow{1}{*}{2.02} \\
\multirow{1}{*}{Example \ref{ex:testout}} & $x \mapsto 0.25 - \|x + 0.5\|_2^2$ &  \multirow{1}{*}{0.25} & \multirow{1}{*}{0.249} & \multirow{1}{*}{0.0993} & \multirow{1}{*}{-0.0777} \\
\hline
\multirow{2}{*}{Example \ref{ex:ahmadi3}} &  $\|\cdot\|_2^2$ & \multirow{1}{*}{10.2} & \multirow{1}{*}{2.84} & \multirow{1}{*}{2.84} & \multirow{1}{*}{2.84} \\
& $x \mapsto \|T^1(x) - T^2(x)\|_2^2$ & 5.66 & 2.81 & 2.78 & 2.78\\
\hline
\end{tabular}
\label{table:benchcmp}
\end{center}
\end{table}
\subsection{Other properties} 
\label{benchsafe}
Here we consider the program given in Example~\ref{ex:testout}. One is interested in showing that the set $X_1$ of possible values taken by the variables of this program does not meet the ball $B$ of center $(-0.5, -0.5)$ and radius $0.5$.
\begin{example}
\label{ex:testout}
Let consider the PPS $\pws=(\xin,X^0,\{X^1,X^2\},\{T^1,T^2\})$ with $\xin = [0.5, 0.7] \times [0.5, 0.7] $, $X^0=\{x\in\rr^2\mid r^0(x)\leq 0\}$ with $r^0:x\mapsto -1$, $X^1=\{x\in\rr^2\mid r^1(x)\leq 0\}$ with $r^1:x\mapsto \| x \|_2^2-1$, $X^2=\{x\in\rr^2\mid r^2(x)\leq 0\}$ with $r^2=-r^1$ and $T^1:(x_1,x_2)\mapsto (x_1^2+ x_2^3,  x_1^3+  x_2^2)$, 
 $T^2:(x,y) \mapsto (0.5 x_1^3 + 0.4 x_2^2, - 0.6 x_1^2 + 0.3 x_2^2)$. With $\kappa : (x_1, x_2) \mapsto 0.25 - (x_1 + 0.5)^2 - (x_2 + 0.5)^2$, one has $B := \{ x \in \rr^2 \mid  0 \leq \kappa(x) \}$. Here, one shall prove $x \in \rea \implies \kappa(x) < 0$ while computing some negative $\alpha$ such that $\rea\subseteq \prop{\kappa, \alpha}$. Note that $\kappa$ is not a norm, by contrast with the previous examples.
\end{example}
At step $m = 3$ (resp.$m = 4$), Program~\eqref{polsynthesis} yields  a nonnegative solution $w_3$ (resp.~$w_4$). Hence, it does not allow to certify that $\rea\cap B$ is empty. This is illustrated in both Figure~\ref{fig:testout} (a) and Figure~\ref{fig:testout} (b), where the light grey region does not avoid the ball $B$. However, solving Program~\eqref{polsynthesis} at step $m = 5$ yields a negative bound $w_5$ together with a certificate that $\rea$ avoids the ball $B$ (see Figure~\ref{fig:testout} (c)). The corresponding values of $w_m$ ($m=3,4,5$) are given in Table~\ref{table:benchcmp}.

\begin{figure}[!ht]
\centering
\subfigure[$m=3$]{
\includegraphics[scale=\sizesmallfig]{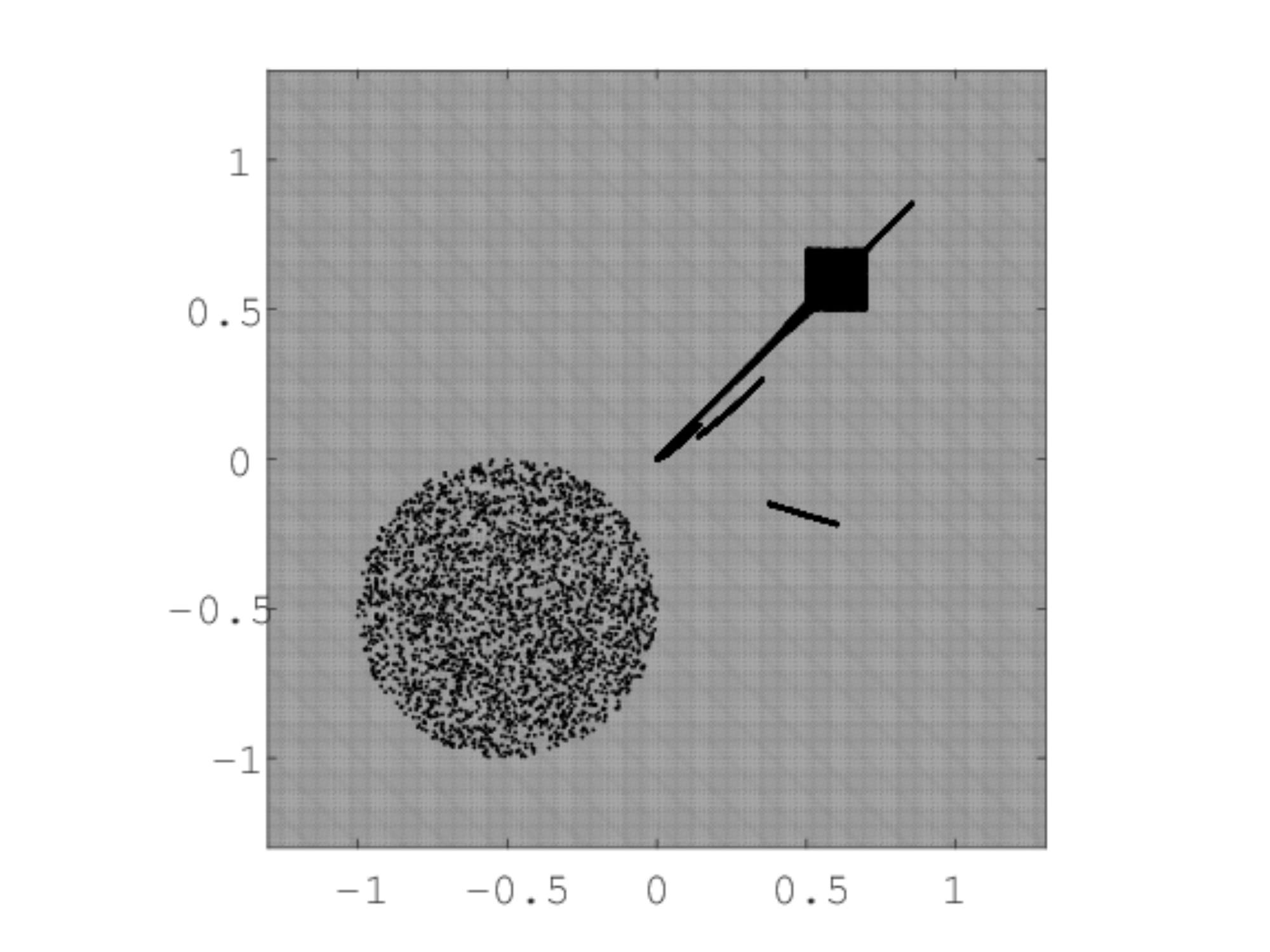}}
\subfigure[$m=4$]{
\includegraphics[scale=\sizesmallfig]{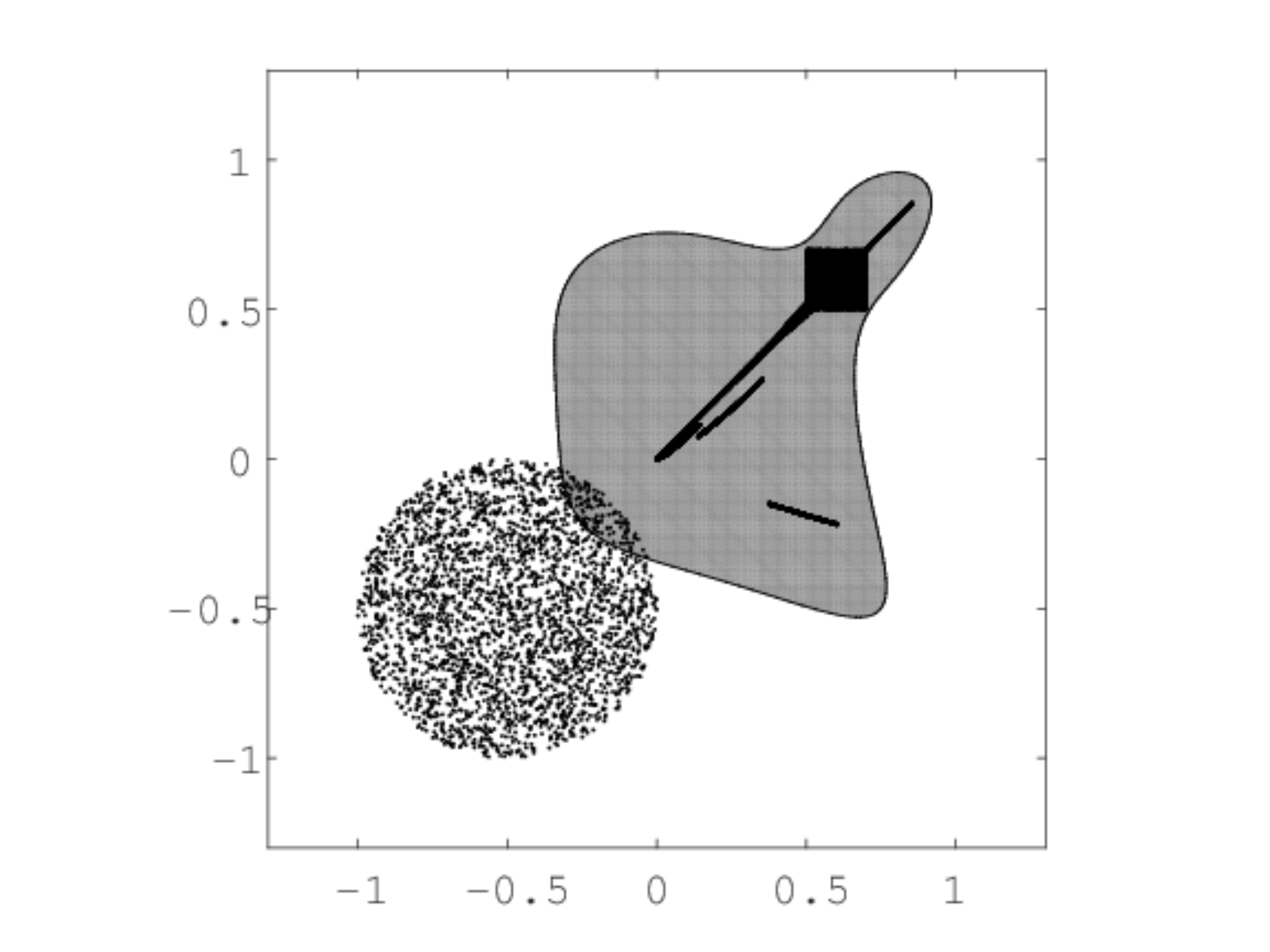}}
\subfigure[$m=5$]{
\includegraphics[scale=\sizesmallfig]{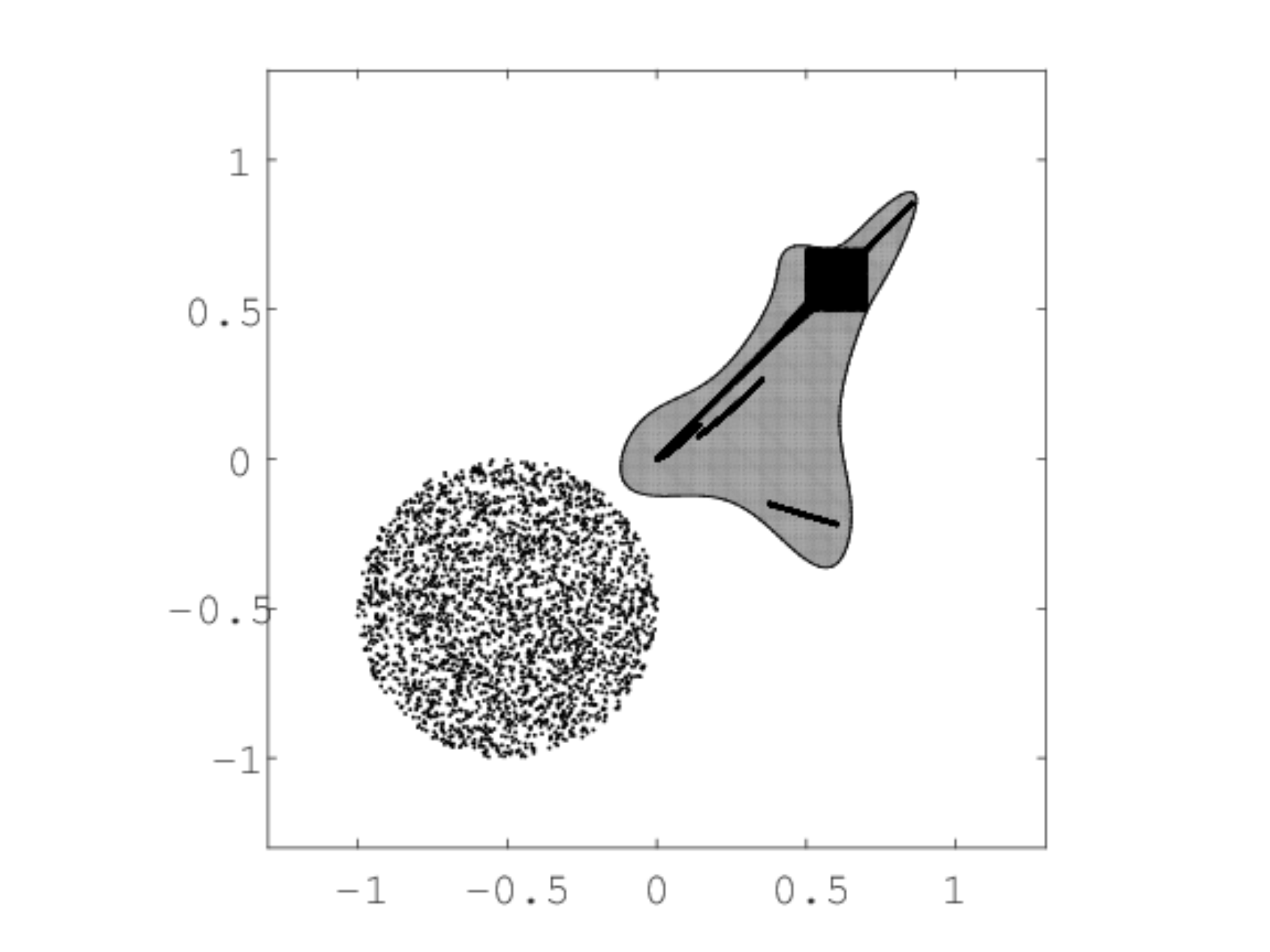}}
\caption{A hierarchy of sublevel sets $P_m$ for Example~\ref{ex:testout}}	
\label{fig:testout}
\end{figure}
Finally, one analyzes the program given in Example~\ref{ex:ahmadi3}.
\begin{example}(adapted from Example 3 in~\cite{ahmad13switched})
\label{ex:ahmadi3}

Let  $\pws$ be the PPS $(\xin,X^0,\{X^1,X^2\},\{T^1,T^2\})$ with $\xin = [-1, 1] \times [-1, 1] $, $X^0=\{x\in\rr^2\mid r^0(x)\leq 0\}$ with $r^0:x\mapsto -1$, $X^1=\{x\in\rr^2\mid r^1(x)\leq 0\}$ with $r^1:x\mapsto x_2 - x_1$, $X^2=\{x\in\rr^2\mid r^2(x)\leq 0\}$ with $r^2=-r^1$ and $T^1:(x_1,x_2)\mapsto (0.687 x_1 + 0.558 x_2 - 0.0001 * x_1 x_2, - 0.292 x_1 + 0.773 x_2)$,  
 $T^2:(x,y) \mapsto (0.369 x_1 + 0.532 x_2 - 0.0001 x_1^2, -1.27 x_1 + 0.12 x_2 -0.0001 x_1 x_2)$. We consider the boundedness property $\kappa_1 := \|\cdot\|_2^2$ as well as $\kappa_2 (x) := \|T^1(x) - T^2(x) \|_2^2$. The function $\kappa_2$ can be viewed as the absolute error made by updating the variable $x$ after a possibly ``wrong'' branching. Such behaviors could occur while computing wrong values for the conditionals (e.g. $r^1$) using floating-point arithmetics. Table~\ref{table:benchcmp} indicates the hierarchy of bounds obtained after solving Problem~\eqref{polsynthesis} with $m=3,4,5$, for both properties. The bound $w_5 = 2.84$ (for $\kappa_1$) implies that the set of reachable values may not be included in the initial set $\xin$. A valid upper bound of the error function $\kappa_2$ is given by $w_5 = 2.78$.
\end{example}

%


\section{Templates bases}
\label{sec:templates}
We finally present further use of the set $P$ defined at Equation~\eqref{eqfondamentale}. This sublevel set can be viewed as a template abstraction, following from the definition in~\cite{DBLP:journals/corr/abs-1111-5223}, with a fixed template basis ${p}$ and an associated $0$ bound.
This representation allows to develop a policy iteration algorithm~\cite{pisos} to obtain more precise inductive invariants.

We now give some simple method to complete this template basis to improve the precision of the bound $w$ found with 
Problem~\eqref{polsynthesis}.

\begin{proposition}[Template basis completions]
\label{onlyone}
Let $(p,w)$ be a solution of Problem~\eqref{polsynthesis} and 
$\mathcal{Q}$ be a finite subset of $\rr[x]$ such that for all $q\in\mathcal{Q}$, $p-q\in\Sigma[x]$.
Then $\rea\subseteq \{x\in\rd \mid p(x)\leq 0,\ q(x)\leq 0,\ \forall\, q\in  \mathcal{Q}\}\subseteq \prop{\kappa,w}\subseteq\prop{\kappa,\alpha}$ and $\{x\in\rd \mid p(x)\leq 0,\ q(x)\leq 0,\ \forall\, q\in  \mathcal{Q}\}$ is an inductive invariant.
\end{proposition}
\begin{proof}
Let $Q$ be the set $\{x\in\rd \mid p(x)\leq 0,\ q(x)\leq 0,\ \forall\, q\in  \mathcal{Q}\}$.
It is obvious that $Q\subseteq P=\{x\in\rd \mid p(x)\leq 0\}$ and hence $Q\subseteq \prop{\kappa,w}$. 
Now let us prove that $Q$ is an inductive invariant. We have to prove that $Q$ satisfies Equation~\eqref{eqinvarianttemplates} 
that is: 
\begin{inparaenum}[(i)]
\item For all $x\in\xin$, $q(x) \leq 0$; 
\item For all $i\in\ind$, for all $x\in Q \cap X^i\cap X^0$,  $q(T^i(x))\leq 0$.
\end{inparaenum} 
For all $q\in\mathcal{Q}$, we denote by $\psi_q$ the element of $\Sigma[x]$ such that $p-q=\psi_q$.  
Let us show (i) and let $x\in\xin$. We have $q(x)=p(x)-\psi_q(x)$ and since $\psi_q\in\Sigma[x]$, we obtain, 
$q(x)\leq p(x)$. Now from Proposition~\ref{sosproposition} and 
Lemma~\ref{lemma:invsufficient} and since $(p,w)$ is a solution of Problem~\eqref{polsynthesis}, we conclude that $q(x)\leq p(x)\leq 0$.

Now let us prove (ii) and let $i\in\ind$ and $x\in Q\cap X^i\cap X^0$. We get $q(T^i(x))=p(T^i(x))-\psi_q(T^i(x))$ and 
since $\psi_q\in\Sigma[x]$, we obtain $q(T^i(x))\leq p(T^i(x))$. Using the fact that $(p,w)$ is a solution of Problem~\eqref{polsynthesis} and using Proposition~\ref{sosproposition} and Lemma~\ref{lemma:invsufficient}, we obtain $q(T^i(x))\leq p(T^i(x))\leq p(x)$. Since
$x\in Q\subseteq P=\{y\in\rd \mid p(y)\leq 0\}$, we conclude that $q(T^i(x))\leq 0$.
\end{proof}
Actually, 
we can weaken the hypothesis of Proposition~\ref{onlyone} to construct an inductive invariant. Indeed, after the computation of $p$ following 
Problem~\eqref{polsynthesis}, it suffices to take a polynomial $q$ such that $p-q\geq 0$. Nevertheless, we cannot compute easily such a polynomial $q$.
By using the hypothesis $p-q\in\Sigma[x]$, we can compute $q$ by sum-of-squares. Proposition~\ref{onlyone} allows to define a simple method to construct a basic semialgebraic inductive invariant set. Then the polynomials describing this basic semialgebraic set defines a new templates basis and this basic semialgebraic set can be used as initialisation of the policy iteration algorithm developed in~\cite{pisos}. Note that the link between the templates generation and the initialisation of policy iteration has been addressed in~\cite{aadje_nsv}.

\begin{example}
  Let us consider the property $\prop{\|\cdot\|_2^2,\infty}$ and let $(p,w)$ be
  a solution of Problem~\eqref{polsynthesis}.  We have $\kappa(x) = \sum_{1
      \leq j \leq k} x_j^2$ and $w + p -\kappa=\psi$ where $\psi\in\Sigma[x]$. In~\cite{DBLP:conf/hybrid/RouxJGF12}, the templates basis used to compute bounds on the reachable values set consists in the square variables plus a Lyapunov function. Let us prove that, in our setting, $\mathcal{Q}=\{x_k^2-w,\ k=1,\ldots,d\}$ can complete $\{p\}$ in the sense of Proposition~\ref{onlyone}. Let $k\in\{1,\ldots,d\}$ and let
  $x\in\rd$, $p(x)-(x_k^2-w)=p(x)-\kappa(x)+w+\sum_{j\neq k}
  x_j^2=\psi(x)+\sum_{j\neq k} x_j^2\in\Sigma[x]$.

\end{example}

\section{Related works and conclusion}

Roux et al.~\cite{DBLP:conf/hybrid/RouxJGF12} provide an automatic method to compute floating-point
certified Lyapunov functions of perturbed affine loop body updates. They use Lyapunov functions with 
squares of coordinate functions as quadratic invariants in case of single loop programs written in 
affine arithmetic. 
In the context of hybrid systems, certified inductive invariants can be computed by 
using SOS approximations of parametric polynomial optimization problems~\cite{Lin14hybrid}.  
In \cite{Prajna04hybrid}, the authors develop a SOS-based methodology to certify that the trajectories of hybrid systems avoid an unsafe region.

In the context of static analysis for semialgebraic programs, the approach developed in~\cite{CousotSDP} focuses on inferring valid loop/conditional invariants for semialgebraic programs\footnote{This approach also handles semialgebraic program termination}. This approach relaxes an invariant generation problem into the resolution of nonlinear matrix inequalities, handled with semidefinite programming.
Our method bears similarities with this approach but we generate a hierarchy of invariants (of increasing degree) with respect to target polynomial properties and restrict ourselves to linear matrix inequality formulations.
In~\cite{BagnaraR-CZ05}, invariants are given by polynomial inequalities (of bounded degree) but the method relies on a reduction to 
linear inequalities (the polyhedra domain).  
Template polyhedra
domains allow to analyze reachability for polynomial systems: 
in~\cite{sassi2012reachability}, the authors propose a method that computes linear templates to improve the accuracy 
of reachable set approximations, whereas the procedure in~\cite{dang2012reachability} relies on Bernstein polynomials and linear programming, with linear templates being fixed in advance. Bernstein polynomials also appear in ~\cite{polynomial_template_domain} as polynomial templates but they are not generated automatically. In~\cite{gulwani}, the authors use SMT-based techniques to automatically generate templates which are defined as formulas built with arbitrary logical structures and predicate conjunctions.
Other reductions to systems of polynomial {\em equalities} (by contrast with polynomial inequalities, as we consider
here) were studied in \cite{Muller,Kapur} and more recently in~\cite{cachera2014inference}.  

In this paper, we give a formal framework to relate the invariant generation 
problem to the property to prove on analyzed program. 
We proposed a practical method to compute such invariants in the case of polynomial arithmetic using sums-of-squares programming. This method is able to handle non trivial examples, as illustrated through the numerical experiments.
Topics of further investigation include refining the invariant bounds generated for 
a specific sublevel property, by applying the policy iteration algorithm. Such a refinement would be of particular interest if one can not decide whether the set of variable values avoids an unsafe region when the bound of the corresponding sums-of-squares program is not accurate enough. For the case of boundedness property, it would allow to decrease the value of the bounds on the variables.
Finally, our method could be generalized to a larger class of programs, involving semialgebraic or transcendental assignments, while applying the same polynomial reduction techniques as in~\cite{AGMW14nltemplates}.
\bibliographystyle{alpha}
\bibliography{sas15mainbib}
\end{document}